\documentclass[a4paper,twocolumn,11pt,accepted=2023-03-01]{quantumarticle}
\pdfoutput=1

\usepackage[numbers,sort&compress]{natbib} 

\usepackage{subfig}
\usepackage{hyperref}
\usepackage{diagbox}
\usepackage{multirow}
\usepackage{float}
\usepackage{qcircuit}
\usepackage{color}
\usepackage{amsmath}
\usepackage{amssymb}
\usepackage{array}
\usepackage{amsthm}
\usepackage{graphicx}
\usepackage{dcolumn}
\usepackage{bm}
\usepackage{soul}
\usepackage[font=small]{caption}
\usepackage{booktabs}
\usepackage{xr}

\newtheorem{theorem}{Theorem}[section]
\newtheorem*{remark*}{Remark}
\newtheorem*{remarks*}{Remarks}
\newtheorem{definition}{Definition}[section]

\newtheorem{lemma}{Lemma}[section]
\newtheorem*{theorem*}{Theorem}
\newtheorem*{lemma*}{Lemma}

\newcommand\bra[1]{\langle{#1}|}
\newcommand\ket[1]{|{#1}\rangle}

\begin{document}
\setlength{\abovedisplayskip}{3pt}
\setlength{\belowdisplayskip}{3pt}

\title{High Dimensional Quantum Machine Learning With Small Quantum Computers
}

\author{S. C. Marshall}
\email{s.c.marshall@liacs.leidenuniv.nl}
\author{C. Gyurik}
\author{V. Dunjko}
\affiliation{%
 Leiden University, Leiden, The Netherlands
}%

\begin{abstract}
Quantum computers hold great promise to enhance machine learning, but their current qubit counts restrict the realisation of this promise. 
To deal with this limitation the community has produced a set of techniques for evaluating large quantum circuits on smaller quantum devices. These techniques work by evaluating many smaller circuits on the smaller machine, that are then combined in a polynomial to replicate the output of the larger machine. This scheme requires more circuit evaluations than are practical for general circuits. However, we investigate the possibility that for certain applications many of these subcircuits are superfluous, and that a much smaller sum is sufficient to estimate the full circuit. We construct a machine learning model that may be capable of approximating the outputs of the larger circuit with much fewer circuit evaluations. We successfully apply our model to the task of digit recognition, using simulated quantum computers much smaller than the data dimension. The model is also applied to the task of approximating a random 10 qubit PQC with simulated access to a 5 qubit computer, even with only relatively modest number of circuits our model provides an accurate approximation of the 10 qubit PQCs output, superior to a neural network attempt. The developed method might be useful for implementing quantum models for larger data throughout the NISQ era.
\end{abstract}

\maketitle

\section{Introduction}

Quantum machine learning is often listed as one of the most promising applications of a near term quantum computer \cite{preskill2018}, with important early successes in a range of problems, from classification \cite{havlivcek2019supervised, schuld2019quantum} to generative modelling \cite{liu2018differentiable}. However the broader roll out of these methods to real world problems is tempered, in part, by the limited size of quantum computers. Among other limitations, current quantum computers lack enough qubits to run large circuits. Some ``circuit partitioning schemes" \cite{bravyi2016, peng2020, mitarai2021constructing} have been proposed to simulate larger circuits on smaller devices by partitioning the full circuit into a set of smaller circuits (see figure \ref{fig:cut up}). However, the exponential number of circuits needed by these schemes is completely intractable for most applications, with billions of sub-circuit evaluations required for even modest quantum machine learning instances.

In this work we examine the necessity of each subcircuit in producing an approximation of some partitioned circuit, presenting reasoning that a smaller amount of circuits could be sufficient in some cases. We then use this as inspiration for a new machine learning technique, which reconciles the need for larger circuit instances with affordable runtimes. Our new technique takes the same form as a given generic machine learning architecture that has been partitioned using the aforementioned techniques but with vastly fewer terms. 

We develop the basic theory behind this technique in Section \ref{Develop Model}, consider its generalisation error in Section \ref{gen error section} and test it experimentally in Section \ref{experiment section} on an instances of handwritten digit recognition using a 64 qubit ansatz with access to only a simulated 8 qubit computer (without use of excessive dimensionality reduction, such as dimensional principal component analysis). We also include an experiment testing the model's ability to replicate the output of larger unpartitioned circuits. Error analysis and the specifics of an evaluation and training schemes are presented in Section \ref{Implementation}.

    \begin{figure}
    \centering
    \includegraphics[width=8cm]{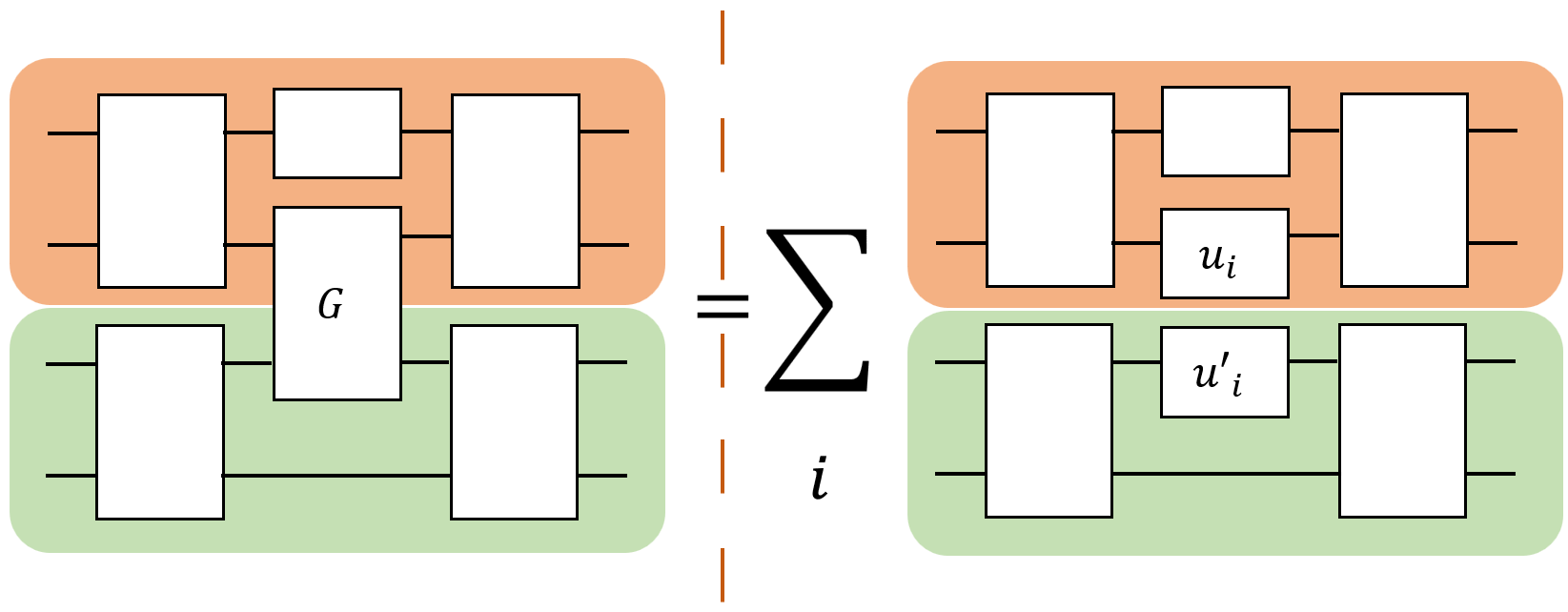}
    \caption{The fundamental notion of circuit partitioning. A potential partition (peach coloured and mint coloured) exists but is joined by a 2-qubit gate, $G$. By expressing $G$ as a sum of single qubit unitaries $u_i$ and $u_i'$ we can simulate the ouput of the large circuit by only running circuits on either element of the partition (which requires a smaller quantum computer).}
    \label{fig:cut up}
    \end{figure}
    
\section{Related work}\label{Rel Work}

We are not the first to consider how the  partitioning schemes \cite{bravyi2016, peng2020, mitarai2021constructing} could be made more efficient, whereas other research lines have focused on minimising the computational cost of applying the \textit{exact} partitioning schemes (e.g. by minimising the number of gates cut), we focus on shrinking the number of subcircuits to \textit{approximate} the output. As such many of the techniques in this section can be composed with our method to create an even more efficient scheme.

In \cite{tang2021cutqc} an automated cutting procedure is applied to \cite{peng2020} to produce the minimum number of subcircuits needed, similarly \cite{perlin2021quantum} uses maximum likelihood fragment tomography to improve both the cutting process and the reconstruction of output states. Other authors have considered how the set of subcircuits could be run more effectively by utilising distributed computational resources \cite{saleem2021quantum, bechtold2021bringing}.

Using partitioning schemes produces an additional benefit: reduced noise, stemming from the smaller circuit size \cite{tang2021cutqc, avron2021quantum}, this can be the motivation for cut selection, even when the full circuit would ``fit" on a quantum machine \cite{basu2021qer}. This noise reduction is similar to the increased accuracy we may be able to provide to gradients in our model. The potential link between this, and the avoidance of barren plateaus in our model we will discuss in Section \ref{Implementation}. 

After developing our technique we will demonstrate its use on high dimensionality data, specifically handwritten digit recognition. This problem has been tackled before with quantum hardware. In \cite{haug2021large}, dimensionality reduction techniques (such as principal component analysis) are used to reduce the dimensionality of the digits to a feature vector small enough to fit on their 8 qubit machine. Similarly, in \cite{li2021quantum} handwritten digits are classified on an 8 qubit machine, in this instance the size of the data is not reduced, the full data is carefully encoded into the quantum computer, first with amplitude encoding, and then by using 11 layers of parameterised gates. Our approach is fundamentally different from either of these. We use the same sized data (8$\times$8 pixels) but do not apply dimensionality reduction as in \cite{haug2021large}, or reuse qubits for multiple data points as in \cite{li2021quantum}. We follow a simple encoding: giving each pixel its own qubit, which we can achieve as we are approximating a 64 qubit machine, while only using an 8 qubit machine.

Other works have addressed high dimensionality data by pushing the limit of the size of quantum machine learning models on current devices \cite{peters2021machine}. Some have employed quantum circuits as components of some larger algorithm to tackle bigger problems: \cite{fujii2022deep} recursively applies PQCs to the outputs of PQCs and \cite{yuan2021quantum} uses quantum circuits as part of a hybrid tensor network but both do not address the task of partitioning a larger circuit and running it efficiently as we do here.

\section{Model Motivation and Specification}\label{Develop Model}
In this section we introduce parameterised quantum circuits (PQCs), a popular concept in quantum machine learning; and circuit partitioning, a method of evaluating quantum circuits that requires a number of qubits greater than what is accessible. By applying these circuit partitioning schemes to PQCs we can produce a more powerful machine learning model than the smaller device naively allows, at the cost of unreasonable runtime. We then go on to develop a novel QML method which intuitively may be as useful requiring only a fraction of the runtime. 

\subsection{Parameterised Quantum Circuits}

Parameterised quantum circuits (PQCs) are a varied and promising method for quantum machine learning. In general they consist of some set of circuits, $\mathbb{U}$, parameterised by a weight vector, $\bm \theta$. In the most common forms the input datum, $\bm x$, also parameterises gates in the circuit. The set can be indexed as $\mathbb{U}=\{ U(\bm x;\bm \theta) \}$. These circuits yield functions when we specify an initial state, $\ket{\phi}$, and an observable, $M$:
\begin{equation}
    f_{\bm \theta, M, \ket{\phi}}(x) = \bra{\phi}U^ \dagger(\bm  \theta ,\bm  x) M U(\bm\theta ,\bm  x)\ket{\phi}
\end{equation}
\noindent
We can assume $\ket{\phi}$ is some fiducial state, such as $\ket{0}^{\otimes n}$ in the computational basis, without loss of generality. 

As each setting of $\bm \theta$ defines a (not necessarily unique) function, $f_{\bm \theta, M, \ket{\phi}}: \bm x \mapsto \mathbb{R}$  (with range limited by the spectrum of $M$), the set of unitaries defines a set of functions. We call this set the {\textit{hypothesis class}}, to coincide with the common usage in machine learning. PQCs have been studied in other contexts, such as quantum chemistry or condensed matter physics \cite{kandala2017hardware}, although it is likely our approaches might generalise to these areas, in this work we focus on its application to machine learning.

\begin{definition}[PQC hypothesis class]\label{model hypothesis class}
The hypothesis class generated by the family of parameterised quantum circuits $\mathbb{U}$ together with an observable $M$ is given by
\begin{align*}
    &\mathcal{F}_{\mathbb{U}, M}=
    \\
    & \hspace{0.8cm} \Big\{f_{\bm \theta}(\bm x) = 
    \\
    &\hspace{1cm}\bra{0}U^ \dagger(\bm \theta ,\bm  x) M U(\bm \theta , \bm  x)\ket{0} : \bm \theta \in [0,2\pi)^{N_{\mathrm{PS}}}\Big\},
\end{align*}

where $N_{\mathrm{PS}}$ is the number of parameters in the model.
\end{definition}

These PQCs have proven popular, but the implementation of PQCs is currently hindered by the NISQ machines they run on. Notably the limited number of qubits available limits the width (defined henceforth as number of qubits the circuit acts on) of the circuit that can be run. It is the central concern of this work to produce a model as useful as PQCs of width larger than what the available machines naively permit.

\subsection{Circuit Partitioning}
In \cite{bravyi2016,peng2020, mitarai2021constructing} the authors propose methods to simulate large quantum circuits on smaller quantum machines by partitioning the circuit into smaller disconnected blocks. In this section we will introduce and then employ these methods on PQCs to decrease the size of quantum computer needed. In our work the decompositions are based on the result of \cite{bravyi2016} however extension to the results of \cite{peng2020, mitarai2021constructing} are also possible. We present only the approach of \cite{bravyi2016} as they are, for our purposes, very similar.

Consider a partition of the $N$ qubits into blocks, $\{B_i\}_i$, where $B_i\subset [N]$ ($[N]:= 1,2,\ldots, N$) such that $\bigcup_iB_i = [N]$ and $B_i \cap B_j = \emptyset$ $\forall i\neq j$ (i.e. each qubit is in one and only one block of the partition). We use the fact that any unitary matrix can be decomposed into a sum of weighted tensor products of single qubit unitaries. In \cite{bravyi2016} this fact is used to decompose any particular 2-qubit gate into a gate of the form:
\begin{equation}
\label{Basic cutting identity}
    U = \sum \alpha_{i} u_i\otimes u'_i
\end{equation}
for some complex $\alpha_{i}$ such that $\sum |\alpha_i|^2 = 1$ and for 2 dimensional unitaries, $u_i$ and $u_i'$. The number of terms of the sum needed for any particular gate is given by its Schmidt number \cite{balakrishnan2011operator}, generically this number is 4 for 2-qubit entangling gates but for some important cases (including the CNOT and controlled-Z) only 2 terms are needed. For example, we can decompose the Controlled-Z gate into single qubit gates as:
\begin{equation}\label{Controlled z decomp}
    \mathrm{Controlled-Z} = \frac{1}{1+i} \left(S \otimes S + i S^\dagger \otimes S^\dagger \right) 
\end{equation}

where $S$ is the phase gate. The identity (\ref{Basic cutting identity}) allows us to rewrite any particular 2-qubit gate as the sum of products of single qubit operators. Applying this method to every 2-qubit gate connecting two blocks of the partition decomposes the full unitary into a sum of tensor products of unitaries which individually act only on each block of the partition (figure \ref{fig:cut up}). 

An example is useful in illustrating this point, suppose we are given a unitary $W$ which consists of two disconnected blocks apart from one 2-qubit gate, $G$, connecting the otherwise disjoint blocks, top and bottom (as in figure \ref{fig:cut up}):
\begin{equation}
    W = (U_\mathrm{top} \otimes U_\mathrm{bot}) G (V_\mathrm{top} \otimes V_\mathrm{bot})
\end{equation}
We can decompose this 2-qubit gate as $G = \sum_i \alpha_{i} u_i\otimes u'_i$. The full unitary can thus be written as:
\begin{equation}
\begin{split}
    W = &(U_\mathrm{top} \otimes U_\mathrm{bot})(\sum_i \alpha_{i} u_i\otimes u'_i)(V_\mathrm{top} \otimes V_\mathrm{bot})\\
      = &\sum_i  \alpha_{i}(U_\mathrm{top} \otimes U_\mathrm{bot})( u_i\otimes u'_i)(V_\mathrm{top} \otimes V_\mathrm{bot})\\
      = &\sum_i  \alpha_{i}(U_\mathrm{top} u_i V_\mathrm{top}) \otimes (U_\mathrm{bot} u'_i V_\mathrm{bot})
\end{split}
\end{equation}
Suppose the initial state is $\ket{0}^{\otimes n}$ (which we will simply refer to as $\ket{0}$)  and the measurement is the projection, $\ket{0}\!\bra{0}$ (using the previous notation for $\ket{0}^{\otimes n}$), we then have that 
\begin{align*}
    \bra{0}&W\ket{0} = 
    \\
    &\hspace{1cm}\bra{0} \sum_i  \alpha_{i}(U_\mathrm{top} u_i V_\mathrm{top}) \otimes (U_\mathrm{bot} u'_i V_\mathrm{bot}) \ket{0} \\
    &\hspace{0.2cm}= \sum_i \alpha_{i}\bra{0}(U_\mathrm{top} u_i V_\mathrm{top})\ket{0}\!\bra{0}(U_\mathrm{bot} u'_i V_\mathrm{bot})\ket{0}
\end{align*}
and the expectation value is given by:
\begin{align*}
    \Big|\bra{0}&W\ket{0}\Big|^2\!=\!
    \\
    &\sum_{i,j}  \bar{\alpha_i}\alpha_{j}\bra{0} (U_\mathrm{top} u_i V_\mathrm{top})^\dagger \ket{0}\!\bra{0} (U_\mathrm{top} u_j V_\mathrm{top}) \ket{0}
    \\
    &\hspace{1.2cm}\cdot \bra{0}(U_\mathrm{bot} u'_i V_\mathrm{bot})^\dagger \ket{0}\!\bra{0} (U_\mathrm{bot} u'_j V_\mathrm{bot}\ket{0}.
\end{align*}
which is the product of inner products local to either element of the partition. This allows us to evaluate each smaller inner product individually and then combine them in a product and sum to replicate the expectation value of the full circuit. Depending on the observable it may be preferable to calculate the expectation value (i.e. the previous equation) or to calculate the inner product presented in the equation before and then square the answer to calculate the expectation value.

These results provide us with a clear path to solve the central goal of this paper thus far, ``How to fit a larger model on a smaller machine". It is simply a matter of specifying a large PQC, then deciding on a partition $\{B_i\}_i$ that separates its initial state and measurement nicely. This partition defines a set of closely related circuits that differ only by the replacement of 2-qubit gates with single qubit gates. The next theorem encapsulates the partitioning of PQCs into a set of a set of smaller subcircuits, and the recombination of them to recreate the result of the larger PQC.

\begin{theorem}[Partitioned model]\label{partitioned model}
For every function $f_{\bm \theta} \in \mathcal{F}_{\mathbb{U}, M}$ and qubit partition $\{B_k\}_{k \in [K]}$ with observable $M = \bigotimes_{k \in [K]}M_k$, there exists a set of coefficients $\{c_i\}_{i \in [T]}$ and unitaries $\tilde{\mathbb{U}} = \{U^{i, k}(\bm \theta, \bm  x), U'^{i, k}(\bm \theta, \bm x)\}_{i \in [T], \text{ } k\in [K]}$ (where each $U^{i,k}$ and $U'^{i,k}$ acts on $n_k$ qubits) which can be combined in a function:
\begin{equation}
\label{Bravyi cut 1}
     \tilde{f}_{\bm \theta}(\bm x) = \sum_{i = 1}^{T} c_i \prod_{k=1}^K \bra{0}U'^{i, k\dagger}(\bm \theta, \bm x)M_{k}U^{i, k}(\bm \theta, \bm x)\ket{0},
\end{equation}
such that $ f_{\bm \theta}(\bm x) = \tilde{f}_{\bm \theta, M}(\bm x)$ for every $\bm \theta, \bm x$. 
For arbitrary gates the number of terms $T$ grows as $16^r$, where r is the number of gates across the partition, but for cut gates with known Schmidt number $T$ is the product of the Schmidt number squared of each cut gate.
\end{theorem}

\begin{remarks*}
In many cases the same subcircuit (or its complex conjugate) appears multiple times in equation \ref{Bravyi cut 1}. By storing its value in classical memory the total number of circuit evaluations can be brought down to $6^r$ where $r$ is the number of gates across the partition (as mentioned in \cite{bravyi2016}). Bounding the number of evaluations needed for a given circuit is a task studied in \cite{mitarai2021overhead} in the context of the scheme in \cite{mitarai2021constructing}.

It must also be noted that Equation \ref{Bravyi cut 1} is composed of inner products, not expectation values, thus requiring 2 circuits to evaluate. Further details on this and the effects of error are considered in Section \ref{Implementation}.
\end{remarks*}

Mapping this theorem onto our example $K=2$, the set of coefficients would be $\{\overline{\alpha}_i \alpha_j\}$ and the set of unitaries would be 
\begin{align*}
    &\big\{ \left\{
U_\mathrm{top}u_iV_\mathrm{top}, 
U_\mathrm{top}u_jV_\mathrm{top}, 
\right\},\\
&\left\{
U_\mathrm{bot}u'_iV_\mathrm{bot},
U_\mathrm{bot}u'_jV_\mathrm{bot}
\right\}\big\}_{i,j}.
\end{align*}

\noindent
This example also illustrates the similarity of terms in equation \ref{Bravyi cut 1}, for every ``top" circuit is identical up to the replacement of $u_i, \space u_j$.

Theorem \ref{partitioned model} is useful to our goal, we can fit any large PQC on a small machine, however we have paid a huge price in the need to run an exponential number of smaller circuits. Indeed given that most QML models are relatively densely connected and increasing depth can lead to improved performance, this exponential overhead in number of cut connections is impractically costly. For example a 2-block division of the hardware efficient ansatz up to depth 6, such as those considered in \cite{jerbi2021variational} to solve a simple task would require over 2 billion distinct sub-circuit evaluations. This rough estimation motivates us to revise our goal to ``how to fit a larger model on a smaller machine in an \textit{acceptable} number of circuit evaluations".

If we are interested in exactly recreating the output of the circuit, this goal might be unattainable, unless we can find an exponential number of terms that perfectly cancel each other. There are fortunately several acceptable simplifications we can make to our goal. First, we are not concerned with the exact replication of the unitary. Since our input states are fixed to $\ket{0}$ we only care about the action of our recreated unitary on this state. Second, we may be content with approximate results, or perhaps even approximate results for \textit{most} input data. Finally, our ultimate goal for a machine learning model, in many cases, is simply to output a binary classifier \cite{mohri2018foundations} (or another simpler discrete set of outputs) so we are not interested in keeping terms which contribute similarly as other terms in the final assignment of a class label.

With this in mind we will now define the subset partition model as the best possible approximation of the full result in theorem \ref{partitioned model} keeping only $L$ terms.

\begin{definition}[Subset partition model]
For a partitioned model, $\tilde{f}_{\bm \theta}(\bm x)$, with set of unitaries $\tilde{\mathbb{U}}$, we define the $L$-subset partition model as a function using the optimal $L$-sized subset of terms $I\subset \tilde{\mathbb{U}}$ given by:
\begin{align*}
     \tilde{f}^I_{\bm \theta}(\bm x) = &\sum_{i \in I} \lambda_i \prod_{k=1}^K \bra{0}U'^{i, k\dagger}(\bm \theta, \bm x)M_{k}U^{i, k}(\bm \theta, \bm x)\ket{0}.
\end{align*}
where we have introduced free parameters $\lambda_i$ that can also be optimised over. In the above definition $I$ and $\bm \lambda$ are optimised to produce the best approximation of $\tilde{f}_{\bm \theta}(\bm x)$, for some given success metric.

\end{definition}

Inner products can often involve complex numbers, therefore the output of the model may be complex. However, most scenarios require a real number, in these cases, we take the real part and discard the imaginary.

This model is a step towards our goal, if we are given the model it would be possible to run some approximation of the partitioned circuit on a small computer in acceptable time. However we lack the capacity to chose the \textit{optimal} set of ``small circuits" $I$, in general choosing this set corresponds to a combinatorial optimisation problem. In the next section we will describe why this problem is challenging and produce a model that can work around it.

\subsection{Reduced Partition Model}
In the last section we tackled the problem of how to fit a large model on a smaller machine, but it required us to run an impractical number of circuits to achieve our goal, we introduced a model to get around this but it was impractical to optimise. We now consider a situation where we are given a runtime ``budget", a hypothetical number of circuits, $L$, that we can afford to evaluate. Choosing which $L$ circuits to evaluate from the set generated by the partitioning to perform optimally is an incredibly challenging combinatoric optimisation problem. This process is additionally complicated by the apparent need to use a quantum computer to assess if the circuits can be ignored. In this section we propose a relaxation of the problem: by parameterising the gates that replaced the 2-qubit gates in the circuit cutting process (henceforth called partition gates) such that all terms in the sum are identical up to these introduced parameters. The problem of optimising circuit selection becomes one of optimising the parameters of the partition gates. 

The first step of this process is parameterising the gates introduced by the partition. There are many options for doing this, for example when cutting the controlled Z we get the decomposition in equation \ref{Controlled z decomp}, replacing the 2-qubit gate on either qubit by $S$ or $S^\dagger$. We then wish to create a new parameterised gate which takes a parameter $\zeta$ such that the parameterised gate is $S$ when $\zeta=0$ and $S^\dagger$ when $\zeta = 1$. $Z^\zeta$ composed with $S^\dagger$ is one choice. Defining $\zeta$ this way also allows us to extrapolate gates for  $\zeta \notin \{0,1\}$, creating a continuous parameter we can use for e.g. gradient descent.

We can use this partitioned-gate-parameterisation trick to replace the set $\{U^{i, k}(\bm \theta,\bm  x), U'^{i, k}(\bm \theta,\bm  x)\}_{i, \text{ } k}$, with a new set, $\{U^{k}(\bm \theta,\bm  x, \zeta)\}_{k\in [K]}$, with just one parameterised unitary for each block of the partition, with different terms of the sum differentiated only by different parameters $\zeta$.

\begin{lemma}
\label{lemma:reduced}
For every $\tilde{f}_{\bm \theta} \in \mathcal{F}^L_{\mathbb{U}, M}$, there exists a set of unitaries $\{U^{k}(\bm \theta, \bm x, \bm \zeta)\}_{k\in [K]}$ and parameters $\bm \lambda, \bm \zeta$ defining a function:
\begin{equation}
\begin{split}
\label{full parameterised Bravyi}
     \bar{f}_{\bm \theta, \bm \zeta, \bm \lambda}(\bm x)  & =
     \\
     &\sum_{i \in [L]} {\lambda_i} \prod_{k\in[K]} \bra{0}U^{k \dagger}(\bm \theta,\bm  x, \bm{\zeta}_{i,k})M_{k}
     \\
     &\hspace{2.7cm}U^{k}(\bm \theta,\bm  x, \bm{\zeta}_{i,k+K})\ket{0},
\end{split}
\end{equation}
\noindent
such that $\tilde{f}_{\bm \theta}(\bm x) = \bar{f}_{\bm \theta, \bm \zeta, \bm \lambda}(\bm x)$ for every $\bm \theta,\bm  x$ and for every observable that can be written as tensor product on the elements of the partition $M = \bigotimes_k M_k$. 
\end{lemma}

In this lemma we have used our new free parameter $\bm \zeta$ to parameterise the partition gates, the parameters needed for these partition gates could be calculated from the partitioning theorem or trained through gradient descent. As mentioned, the advantage of this step is that now all terms of the sum are equivalent to each other up to weight parameters $\lambda_i$ and $\bm \zeta_{i,k}$. This is useful in the final model, where we reduce the number of terms to $L$ and then allow these parameters to learn freely, making the model capable of replicating any $L$ terms present in the original model by changing $\lambda_i$ and $\bm \zeta_{i,k}$.

\begin{definition}[Reduced partition model]
For a PQC hypothesis class $\mathcal{F}_{U, M}$, we define the reduced $L$-subset partition model as the family of functions $\overline{\mathcal{F}^L_{\mathbb{U}, M}} = \big\{ \bar{f}_{\bm \theta, \bm \zeta, \bm \lambda}\big\}$ where each function is given by
\begin{equation}
\label{RPM sum}
\begin{split}
    \bar{f}_{\bm \theta, \bm \zeta, \bm \lambda}(\bm x) &= 
    \\
    &\sum_{i \in [L]} \lambda_i \prod_{k\in[K]} \bra{0}U^{k \dagger}(\bm \theta,\bm  x, \zeta_{i,k})M_{k}
    \\
    &\hspace{2.7cm}U^{k}(\bm \theta, \bm x, \zeta_{i,k+K})\ket{0}
\end{split}
\end{equation}
where the unitaries $U^{k}$ are those described in Lemma~\ref{lemma:reduced} and we have introduced entirely free parameters $\bm \lambda$ and $\bm \zeta$ that can be optimised over. This can also be referred to as a ``reduced partition model" when $L$ is to be specified later.
\end{definition}


This new model introduces more free parameters, $\bm \zeta$, into our model, fortunately only $2L\times$ \textit{number of cut gates} are introduced.

The reduced partition model can now use the similarity of the terms of the equation \ref{RPM sum} (they are identical up to the weight vector, $\bm \zeta$) to replicate any subset of terms taken from the partitioned model by simply adjusting the parameters $\bm \lambda$ and $\bm \zeta$. This is stated formally in the following theorem.

\begin{theorem}\label{thm:inclusion}
For any PQC hypothesis class $\mathcal{F}_{\mathbb{U}, M}$, the $L$-subset partition hypothesis class $\mathcal{F}^L_{\mathbb{U}, M}$ is included in the hypothesis class of the reduced $L$-subset partition model $\overline{\mathcal{F}^L_{\mathbb{U}, M}}$, i.e., 
\begin{equation}
    \mathcal{F}^L_{\mathbb{U}, M} \subset \overline{\mathcal{F}^L_{\mathbb{U}, M}}.
\end{equation}
\end{theorem}

In other words, if a given classifier can be sufficiently well approximated by considering only $L$ terms, then the hypothesis class of the reduced partition model can do at least as well as this approximation. This is the potential advantage of our model. Additionally, the relaxation from manually picking terms to optimising $\bm \zeta$ allows us to apply gradient based methods, generally yielding much easier optimisation, but in general suffers as the solutions of the relaxation do not encode meaningful solutions of the original problem (which is discrete in nature). 
However in our case, since we deal with QML, all this achieves is expanding the hypothesis class, where any solution is meaningful, and optimisation (if done completely) can only yield better results with respect to the training error. Although, when expanding the hypothesis class, the problem may become worse generalisation performance, often evidenced by looser/worse generalisation bounds. We analyse these in the next section.


\section{Generalisation Error}\label{gen error section}

In creating the reduced partition model, we partitioned the circuit and removed terms, which intuitively makes the model simpler. But then, we introduced free parameters, making the model more complicated and increasing the size of the (reduced model) hypothesis class. 
We will formally study the effect these alterations have in terms of \textit{generalisation error}, defined roughly as the gap between performance on a training set and performance on unseen data from the same distribution. We will only briefly define a few concepts that are needed, readers keen to see a more are referred to \cite{mohri2018foundations}.

We define a supervised learning task on a domain $\mathcal{X}$ and co-domain $\mathcal{Y}$ with a probability distribution over $\mathcal{X}\times \mathcal{Y}$, $P$, and loss function, $\ell:\mathcal{Y}\times \mathcal{Y} \rightarrow \mathbb{R}$. Supervised learning is the task of outputting a hypothesis, $h \in \mathcal{Y}^\mathcal{X}$, such that the risk, $R(h)$ is minimised. We define the risk for a hypothesis $h$ on a continuous space as the expected loss:
\begin{equation}
    R(h) = \int_{\mathcal{X}\times\mathcal{Y}} \ell(h(x), y) dP(x,y).
\end{equation}
In practical settings we normally lack access to the underlying probability distribution, so the true risk cannot be evaluated. Instead we are supplied with training data drawn from $P$, $S = \{(x_i,y_i) \sim P\,|\,i\in [m]\}$, and must settle for evaluating the risk on this finite set. We call this the empirical risk of $h$ with respect to $S$:
\begin{equation}
    \hat{R}_S(h) = \frac{1}{|S|}\sum_{(x_i,y_i) \in S} \ell(h(x_i), y_i).
\end{equation}
Optimising our hypothesis on the training data optimises the empirical risk, which is generally a good proxy for the true risk. The gap between these two risks is bounded by generalisation bounds, specifically by a generalisation gap function $g$, which can depend on many properties given the specifics of the learning task and hypothesis class. 
Here we will bound the gap with a function independent of the data distribution, depending only on the the hypothesis class, $\mathcal{F}$, the size of the training set, $m$, and an acceptable overall failure probability, $\delta$. 

More precisely, we will aim for a probabilistic bound on generalisation gap in terms of  the risks of the form:
\begin{equation}
    P\left(R(h)<\hat{R}_S(h) + g(\mathcal{F}, m, \delta)\right)>1-\delta
\end{equation}
Intuitively, the gap has to do with the concept of ``overfitting". Simple models tend to have much smaller generalisation gaps. A function which outputs random labels and does no learning has a generalisation gap of $0$ but a large empirical risk. In contrast, some complex and large models are found to ``overfit" data, where the empirical risk drops to near zero but the generalisation of the model is very poor, with poor performance on data points not seen during training. Since our model contains more parameters than the model we derived it from, we might fear we have slid into the poor generalisation-good empirical risk category. 

We study this question, first, by using a well known result on the generalisation performance of quantum circuits, using generalised trigonometric polynomials \cite{caro2021encoding} and then by using the additive property of the Rademacher complexity to examine the effect $L$ has on generalisation performance. In the first analysis we will see that the method does not distinguish the generalisation properties of our model from that of the unpartitioned PQC it was derived from. One one hand this is positive, the new parameters did not decrease performance according to this bound, but on the other hand this is negative as these bounds do not identify $L$ as a parameter influencing model complexity at all, which intuitively it should. To overcome the latter issue, the second analysis uses the additive property of the Rademacher complexity to examine how more terms (a larger $L$) in our model effects the generalisation error. 

\subsection{Encoding Dependent Generalisation Gap} \label{sect:Encoding Dependent Generalisation Gap}

Recall, our objective is to study the generalisation bonds of our model, which attains the form in equation \ref{RPM sum}, where the salient paramters are the number of terms in the summand ($L$) and the number of blocks in the product ($K$).

One insightful analysis of generalisation performance is given in \cite{caro2021encoding}, we will show that its bounds apply directly to our model. The analysis first imports a result shown in \cite{schuld2021effect, gil2020input}, that the  function implemented by any PQC, $f_{\bm \theta} (x)$, is as a generalised trigonometric polynomial(GTP):
\begin{equation}\label{Gtp basic}
    f_{\bm \theta} (x)=\sum_{\omega \in \Omega} c_\omega(\bm \theta, M) e^{-i\omega x}
\end{equation}
\noindent
where the effect of all the parameterised gates and the measurement is only reflected in the coefficients $\{c_\omega\}$. The frequencies available in the GTP ($\Omega$) are determined entirely by the input data's encoding strategy, specifically the eigenvalue spectra of Hamiltonians encoding the input data, typically as rotation gates. Further study of the spectra of frequencies, $\Omega$, is available in the aforementioned works.

With a very similar analysis it can also be shown that a GTP of this form exists for each term of our sum: Consider a single term, $T_i$
\begin{align}
     &T_i =  
     \\&\lambda_i \prod_{k\in[K]} \bra{\phi_k}U^{k}(\bm \theta, \bm x,\bm  \zeta_{i,k})M_{k}U^{k}(\bm \theta, \bm x,\bm  \zeta_{i,k+K})\ket{\phi_k}
\end{align}
\noindent
this is equivalent to reuniting $\ket{\phi_k}$ and $M_k$ from product form, and combining $\bigotimes_{k\in[K]}U^{k}(\bm \theta,\bm  x, \zeta_{i,k+K}) =: U(\bm \theta, x, \bm \zeta_{i})$ into:
\begin{equation}
     T_i =  \lambda_i  \bra{\phi}U(\bm \theta, \bm x, \bm \zeta_{i})M U(\bm \theta, \bm x, \bm \zeta'_{i})\ket{\phi}
\end{equation}
\noindent
this term is now an inner product of an incredibly similar form to the PQC it is derived from, indeed if the encoding gates are untouched by the partitioning scheme then $T_i$ has the same encoding gates and it can be shown admits a representation as a GTP of the same form, with the exact same spectra, $\Omega$. Our new GTP will contain different (and now possibly complex) $\{c_\omega\}$.

Since each term can be represented as a GTP with the same $\Omega$ we are able to combine them into another GTP: 
\begin{align*}
     &\bar{f}_{\bm \theta, \bm \zeta, \bm \lambda}(\bm x)= 
     \sum_{i \in [L]} \lambda_i \sum_{\omega \in \Omega} c_\omega(\bm \theta, \zeta_i,M) e^{-i\omega \bm x}=\\
     &\sum_{\omega \in \Omega} e^{-i\omega \bm x} \sum_{i \in [L]} \lambda_i c_\omega(\bm \theta, \bm \zeta_i,M) =
     \\
     &\sum_{\omega \in \Omega} e^{-i\omega \bm x} c'_\omega(\bm \theta, \bm \zeta,M)
\end{align*}
\noindent
with new weights: $\{c_\omega' =  \sum_{i \in [L]} \lambda_i c_\omega(\bm \theta, \zeta_i,M)'\}$.

This defines a new GTP of the same degree and the same $\Omega$ as the full sized circuit which we originally partitioned. Performing the analysis of type presented in \cite{caro2021encoding} for our circuit gives identical bounds as for the whole (unpartitioned) model.

As our model dramatically differs in the number of terms (which ought to decrease the gap), yet is much more general in the parameters that are free (which should increase the complexity), we see that this bounding technique is quite coarse-grained. In particular, even just pure product models (no entangling gates) which are trivially classically simulatable have the same bounds. The fact that the GTP approach yields somewhat loose bounds was emphasized in \cite{caro2021encoding} and as we have only bounded that bound we must tighten our analysis to achieve a meaningful bound, in the next section we will achieve this.

\subsection{Term-Based Generalisation Gap}\label{term gen error}
In subsection \ref{sect:Encoding Dependent Generalisation Gap} we saw that 
using the analysis technique from \cite{caro2021encoding} the generalisation error analysis for the un-partitioned model matched those of our new model. 

The reason for this was that this method inherently only analyzes the way the data is encoded (i.e., how the individual unitaries depend on the input), and this feature is not different between the partitioned and unpartitioned model.
In order to obtain bounds which are actually sensitive to the cutting process it is important to examine another approach. We want to analyse an approach that fundamentally considers the increasing number of terms. 
To this end, in this section we will introduce and bound a complexity measure, the Rademacher complexity, finding that our bound scales  linearly in $L$.

The Rademacher complexity \cite{mohri2018foundations} is measure of a function family's ability to assign arbitrary labelling to a set of input data. Given some particular input dataset $S=(x_1,x_2,\ldots, x_m)$ the Rademacher complexity of a function family $\mathcal{F}$ is 
$$
\mathcal{R}_S(\mathcal{F}) = \frac{1}{m} \mathbb{E}_\sigma\left[
\sup_{f\in\mathcal{F}}\sum_{i=1}^m \sigma_i f(x_i)
\right].
$$
The above expectation is over $m$ i.i.d. binary random variables $\sigma$ with equal chance of being $+1$ or $-1$.
The random variables simulate the random labeling of the data. By maximizing $\sum_{i=1}^m \sigma_i f(x_i)$ we identify the classifier, element of the hypothesis class, which intuitively  does the best job of matching these random labels on average. Our results hold for any given dataset so we will omit $S$ and use just $\mathcal{R}(\mathcal{F})$

The main tool we will use is the sub-additive property of the Rademacher complexity \cite{mohri2018foundations, bartlett2002rademacher}:
\begin{align}
    \label{eq:rademacher}
    \mathcal{R}\big(\mathcal{F} + \mathcal{G}\big) \leq \mathcal{R}\big(\mathcal{F}\big) + \mathcal{R}\big(\mathcal{G}\big).
\end{align}
For two families of functions $\mathcal{F}$ and $\mathcal{G}$, where the sum $\mathcal{F}+\mathcal{G} = \{f+g:f \in \mathcal{F}, g \in \mathcal{G}\}$. Taking $R_T$ as the maximum Rademacher complexity of any of the summands in our model. We can bound the Rademacher complexity of our model by  $O(R_T L)$, a linear increase with the number of terms in our model. 
This bound does not directly depend on $K$ (the number of partitioned blocks in our model). However, larger values of $K$ may require larger values of $L$ in order to mimic the behaviour of the unpartitioned model.

Comparing $R_T$, the Rademacher complexity of a single term of the summand in the model, to the Rademacher complexity of the unpartitioned model presents challenges: In general, a single term of the model is a product of smaller blocks, where the two qubit gates between blocks have been replaced by single qubit gates. Intuition tells us that removing these connecting gates from a circuit should reduce the expressivity, however, cases can be constructed where removing two qubit gates increases generalisation performance (an example is presented in appendix \ref{appendix}). Thus it is impossible to simplify the Rademacher bound any further while remaining maximally general i.e. without resorting to circuit specific methods.
\section{Evaluation and training of reduced partition model}\label{Implementation}

In this section we look at how one can evaluate the circuits, what error this would entail, and how it might be trained, we speculate on a possible feature of partitioned PQCs that might placate the effects of so-called ``barren plateaus".

\subsection{Evaluation} 

Evaluation of the reduced partition model is a non-trivial task, the terms are composed not of expectation values (which can be evaluated with simple circuits) but of inner products, with different unitaries on either side of the observable. Fortunately this is not an insurmountable problem. To evaluate these inner products we can employ the Hadamard test shown in figure \ref{fig:hadamard}. The most challenging component of this circuit is the application of a controlled-$U/V$, naively this would require controlled gates for every gate in $U$ and in $V$. Fortunately this is not the case. Since $U$ and $V$ differ only by the partition gates, the controlled-circuits can be constructed with controlled operations only on these partition gates, which is a small subset of the total number of gates in the circuit.

\begin{figure}
    \centering
    \includegraphics[width=6cm]{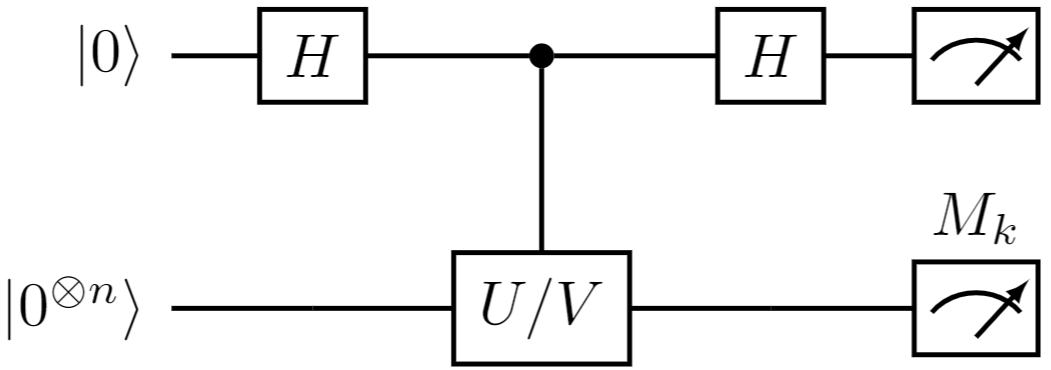}
    \caption{The hadamard test to be used for calculating the real component of some $\bra{0}UM_kV\ket{0}$. The controlled $U/V$ circuit implements $U$ if the control is 0 and $V$ if the control is 1. It is important to note that the controlled $U/V$ circuit only requires on control on a few gates, since only the partition gates differentiate $U$ and $V$ most gates are identical and do not require control. This circuit can be modified to calculate the imaginary component as described in \cite{bravyi2016}}
    \label{fig:hadamard}
\end{figure}

As with all NISQ applications we must inspect how our algorithm will perform on a noisy device. By bounding the variance we find the noise scaling very reasonable.

First, let us replace the non-random inner products, $\bra{\phi_k}U^{k}(\bm \theta, \bm x, \bm \zeta_{i,k})M_{k}U^{k}(\bm \theta,\bm  x, \bm \zeta_{i,k+K})\ket{\phi_k}$, with random variables $X_{i,k}$ which are unbiased estimators of the inner product (that is $\bra{\phi_k}U^{k}(\bm \theta, \bm x, \bm \zeta_{i,k})M_{k}U^{k}(\bm \theta, \bm x, \zeta_{i,k+K})\ket{\phi_k} = \overline{X_{i,k}}$ where the bar now represents the expectation value). These random variables represent an estimation of the inner product with $s$ shots on a quantum computer. 
We are interested in bounding the probability that the difference between the estimate and the average exceeds some $\epsilon$ by $\delta$.
\begin{equation}\label{target error bound}
   P\left(\left|  \sum_{i \in [L]} \lambda_i \prod_{k\in[K]} X_{i,k} - \sum_{i \in [L]} \lambda_i \prod_{k\in[K]} \overline{X_{i,k}}  \right|>\epsilon \right) \leq \delta 
\end{equation}
We can achieve this bound by considering the variance. We will assume the observable and each $\lambda$ are bounded by 1, although we will comment on how this is easily generalised. We find the variance scales with the number of shots:
$$
\sigma^2 \leq  \frac{4K^2L}{s}
$$

by \cite{286787}. Equation \ref{target error bound} is satisfied when we have $s = \frac{4LK^2}{\epsilon^2 \delta}$ shots by Chebyshev's inequality. To generalise this to an observable or variance greater than one, note that the argument of the probability in equation \ref{target error bound} can simply be re-scaled as both of these elements are linear, this in-turn re-scales the variance providing a bound.
\subsection{Training}

Training with a gradient based approach is easy to apply in our model too. The derivative distributes on terms of the sum and can be evaluated by applying the chain rule to the product in each term. Indeed since most parameters appear in only one gate on one qubit on one side of the partition, the chain rule evaluates to 0 on all but 1 element of the product. Evaluating the gradient then takes at most $L$ times the number of evaluations required to evaluate the gradient of one of the smaller circuits. In this case we find that evaluating the gradient for any parameter, $\bm \theta^t$, that exists only in the $k'$th partition is:
\begin{equation}
\begin{split}
\label{grad}
     &\frac{\partial}{\partial \bm \theta^{t}}\bar{f}_{\bm \theta, \bm \zeta, \bm \lambda}(\bm x) = 
     \\&\sum_{i \in [L]} \lambda_i \frac{\partial}{\partial \bm \theta^{t}}\prod_{k\in[K]} \bra{\phi_k}U^{k}(\bm \theta, \bm  x,  \bm \zeta_{i,k})M_{k}
     \\
     &\hspace{4cm}U^{k}(\bm \theta,\bm  x,\bm  \zeta_{i,k+K})\ket{\phi_k}=
     \\&\sum_{i \in [L]} \lambda_i \frac{\partial}{\partial \bm \theta^{t}} \bra{\phi_{k'}}U^{k'}(\bm \theta,\bm  x, \bm \zeta_{i,k'})M_{k'}
     \\
     &\hspace{4cm}U^{k'}(\bm \theta, \bm x, \bm \zeta_{i,k'+K})\ket{\phi_{k'}} \times
     \\&\prod_{k\in[K]\setminus k'} \bra{\phi_k}U^{k}(\bm \theta, \bm x, \bm \zeta_{i,k})M_{k}U^{k}(\bm \theta, \bm x, \bm \zeta_{i,k+K})\ket{\phi_k}
\end{split}
\end{equation}
The same applies for the $\zeta$ parameters. In many instances the gradient can be made easier to compute, since we have often already evaluated the non-derivative expression before looking for the gradient most of the circuit evaluations are already done, with only the derivative expression for a single inner product requiring a new evaluation. Which can be done in the standard manner (e.g. parameter shift rule \cite{crooks2019gradients}).

\subsection{Barren Plateaus}

A well studied problem \cite{mcclean2018barren} with PQCs is the ``barren plateaus" phenomenon, where large parts of the parameter landscape have an exponentially small gradient, effectively crippling optimisation. This is a manageable problem for currently implementable PQCs due to their limited size, but as PQCs become larger (and their gradient decreases) the problem intensifies \cite{mcclean2018barren}. While our model is not immune to barren plateaus we may be able to reduce their effect on our model relative to the size of their effect on the unpartitioned circuit.

Each term of our model is a multiplicative separable function (it can be written as: $f_T(x_1,\ldots ,x_K) = f_1(x_1)\times \ldots \times f_K(x_K)$, where $f_i$ is an inner product and, $x_i$ is the input to the inner product, including the data and weights) we simplify to assuming $x$ is a single parameter, for illustrative purposes. To calculate the gradient we apply the chain rule to the product, for most architectures any particular parameter will only appear on one block of the partition, then one term of the chain rule will be non zero 
$\partial_{x_i}f_T(x_1,\ldots ,x_K) = f_1(x_1) \ldots \left(\partial_{x_i}f_i(x_i)\right)  \ldots f_K(x_K)$. 

The gradient is thus determined by multiplying together the many amplitudes stemming from the subcircuits of the sum with this lone gradient term (equation \ref{grad}). Two aspects may make this overall gradient small: First the gradient may be small as it is a PQC and is prone to barren plateaus, however the individual subcircuits generically have larger gradient than the full unpartitioned circuit as they are smaller \cite{mcclean2018barren} (i.e., the barrenness of the plateaus heavily depends on the number of qubits in the circuit). Second, the multiplication with other terms may cause it to decay to zero as we are dealing with a product of terms which are absolute value below 1, the product then decays exponentially in the number of multiplicative terms to some small number. However in our case we are not directly facing this radically smaller number, we fundamentally have more information about the gradient, knowing the total gradient, but also the terms that are combined to form it. We know the effect that varying any of these subterms has on the gradient of the complete circuit. One possible use of this information is to identify which term is driving the gradient to a small value, and to revert its parameters back to an earlier instance which we have stored in memory, through this method the impact of barren plateaus could be mitigated. A technique similar to \cite{sack2022avoiding} could be developed, to avoid low gradient directions, but utilising the more information present in our case.

There is quite a bit of research on \textit{additive} separable functions, which may transfer to our case \cite{wright2015coordinate}. This could lead to significantly easier training. We plan to develop this method of training in a follow-up work.

\section{Numerics}\label{experiment section}

\setlength\aboverulesep{0pt}
\setlength\belowrulesep{0pt}

\begin{table}
\begin{center}
\begin{tabular}{|m{7em}|m{5em}|m{6em}|}
    \hline
    Experiment&Model \newline Type&Accuracy(A)/\newline MSE \\
    \hline
    MNIST & Neural Net&100\%A \\\cmidrule{2-3}
    handwriting& RPM & 96.4\%A
    \\
    \hline
    Approximating\newline larger PQC & Neural Net RPM&0.424MSE 0.0322MSE\\
    \hline 
\end{tabular}

\end{center}
\caption{A summary of the main numerical findings in this section, we report our model's (RPM) performance on handwriting recognition and on simulating the output of a larger PQC. We also train a neural network as a comparison.}
\end{table}

In the previous sections we laid out a model with considered theoretical underpinning, in this section we will demonstrate that model's basic utility by showing it can learn a simple large problem, the MNIST handwritten digit recognition, by utilising an ansatz much larger than the computer it has simulated access to. We also present an experiment designed to test if an adequate approximation of a random circuit output can be made with much fewer terms, we then apply our model on the same random circuits output to test its performance on synthetic data.
    
\subsection{A Large Problem: Handwriting}

Reading handwritten numbers is one of the most basic tasks in undergraduate machine learning courses. The MNIST \cite{deng2012mnist} data set presents a relatively simple task, identify which digit is written in an $28\times 28$ pixel image, but even this simple task is difficult for current generation quantum machines due to its high dimensionality, with quantum attempts only succeeding recently through careful encoding of the problem (e.g. in \cite{li2021quantum}). Often dimensionality reduction techniques such as principal component analysis are applied \cite{2021data} but for a simple problem like MNIST handwriting this reduces the learning problem to a triviality. Here we preserve the learning problem by downsampling the image to just 64 pixels, which is importantly still human readable. Here we will show that even simple cases of our model perform adequately and by increasing $L$ (the number of terms of our model) we increase that performance. 

\begin{figure}
    \centering
    \includegraphics[width=5cm]{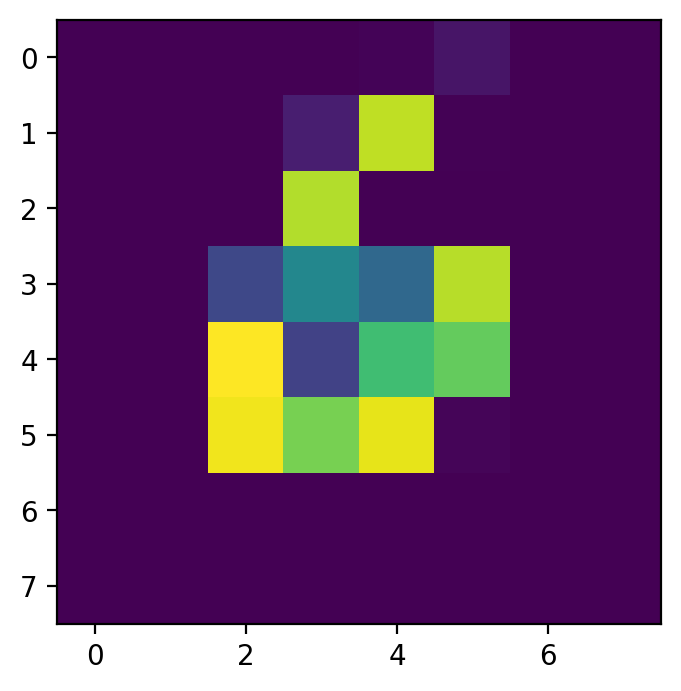}
    \caption{An example datum of the handwriting task, the number 6. The picture is then cut into 8 elements (as given by its columns) with each element as input to a different PQC. This resolution data was chosen so as to be fine enough to be human readable.}
    \label{fig:clearly 6}
\end{figure}

For purposes of comparison we reduce the problem to differentiating 3 and 6, as in \cite{farhi2018classification}. Our model is based on an 8 block partitioning of the 64 qubit, depth 3 hardware efficient ansatz (of the same form as in \cite{jerbi2021variational}) into 8 qubit blocks, a model which would normally be far outside of our computational power. An unseen validation set is evaluated at every step of training and the results are shown in figure \ref{fig:val loss}. The final training loss (MSE), testing loss (MSE) are shown in the table \ref{tab:handwriting}, we also apply a step function to the output (to convert its real valued output into a binary label) and list its accuracy.
\begin{table}
    \begin{center} 
        \begin{tabular}{c|c|c|c}  
            L & loss$_{training}$ & loss$_{testing}$ & acc\% \\
            \hline
            1  & 0.0521 & 0.0499    & 92.8\\
            3  & 0.0488 & 0.0465    & 92.8\\
            5  & 0.0479 & 0.0454    & 94.7\\
            10 & 0.0432 & 0.0422    & 96.2\\
            20 & 0.0359 & 0.0341    & 96.4\\
            \hline \\
            Neural Net &0.0025 & 0.0031 & 100 \\
        \end{tabular}
    \end{center}
    
    \caption{Final loss on training/validation sets and ultimate accuracy of the reduced partition model on the handwriting recognition task. Different values of the hyperparameter $L$ are listed. A neural network is also trained on the task and found to perform very well.}
    \label{tab:handwriting}
\end{table}

Data augmentation (skews and rotations) were used to generate more data for the model. Without this augmentation high $L$ terms began to overfit, increasing the training performance while decreasing the validation performance. With data augmentation we can see that our model is behaving well, even in the 1 term case we find that it selects a good arrangement of weights, although with relatively few additional terms the performance increases, for contrast to run this model using the complete circuit partitioning scheme would require the evaluation of over 46000 subcircuits. A neural network with a convolution layer and a single dense 128 neuron hidden layer is provided for comparison. We must consider that our results are on MNIST handwriting, which is known to have many problems and cannot be used to claim that our model excels on all similarly large tasks \cite{tweet}. 

    \begin{figure}
    \centering
    \includegraphics[width=8cm]{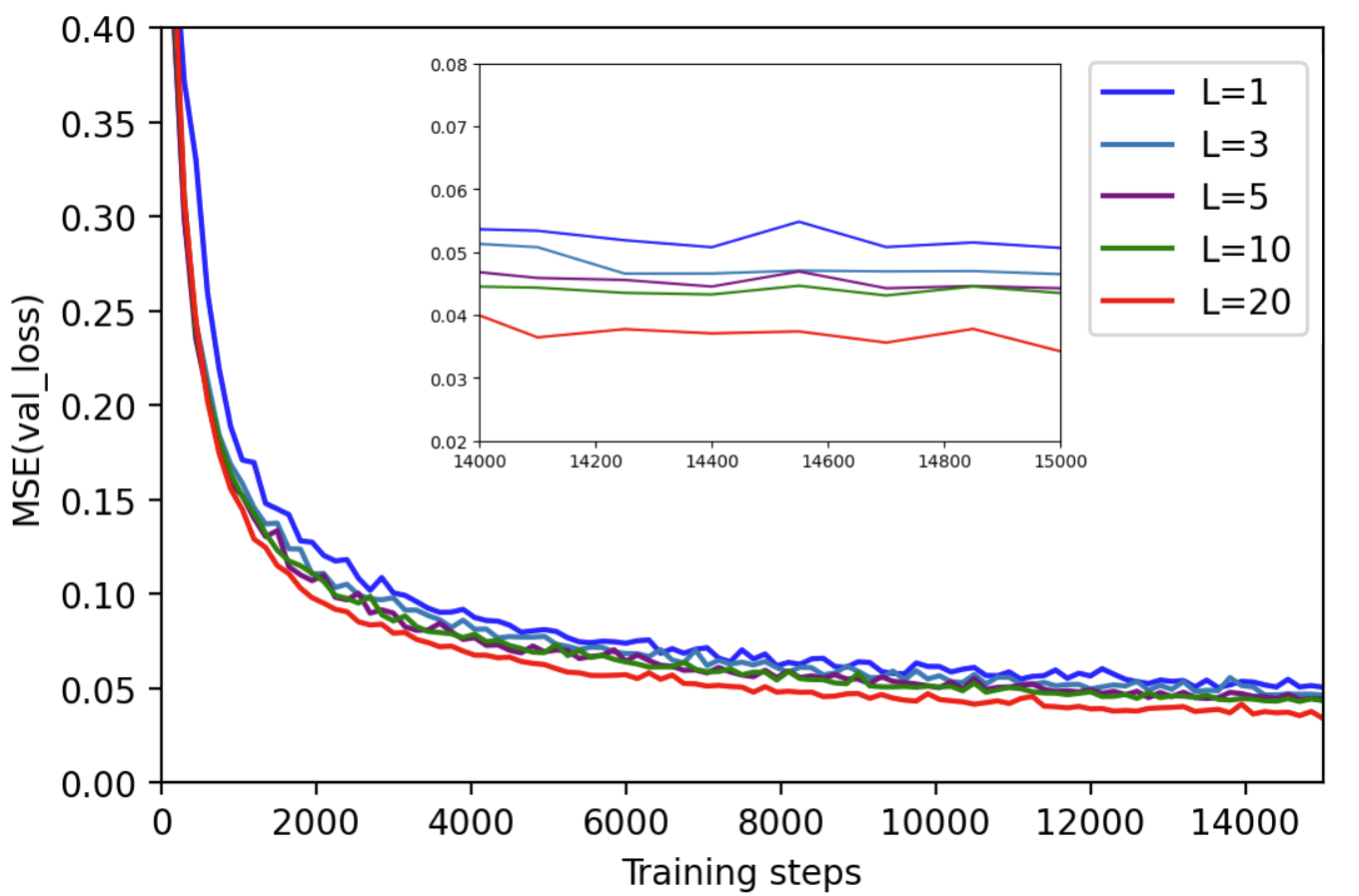}
    \caption{The loss on an unseen data set evaluated alongside the training of a reduced partition model to recognise handwritten digits. Increasing the number of terms has a positive effect on the ability of the model.}
    \label{fig:val loss}
    \end{figure}

\subsection{Tests on Synthetic Data}

This work is built around the assumption that many terms in the partitioned equation for a given circuit are redundant, and a good approximation of the complete circuit can be made by our model. In this section we test this assumption with our first experiment, and then test our complete model on learning a synthetic data set in the second and third experiments.

We take a width 10, depth 3 instance of the hardware efficient ansatz with random weights. Using this circuit we generate a synthetic data set by recording its output on $10000$ random inputs, we normalise these outputs to to a mean squared average of 1. We then instantiate a modified version of our model corresponding to a partitioning of the full circuit into 2 blocks of width 5. The model is modified from the general model we have described above by fixing $\bm \theta$ (the weights present from the unpartitioned PQC) and only training $\bm \zeta$ and $\bm \lambda$ (the weights we introduced when creating the model). This modification allows us to examine directly our claim that introduction of the free parameters, $\bm \zeta$ and $\bm \lambda$, is sufficient to approximate the output of the full PQC without evaluating the many subcircuits that would be required in theorem \ref{partitioned model}. After this experiment we free $\theta$ (apply the full model) and examine the increased performance this gives us.

\begin{figure}
    \centering
    \includegraphics[width=8cm]{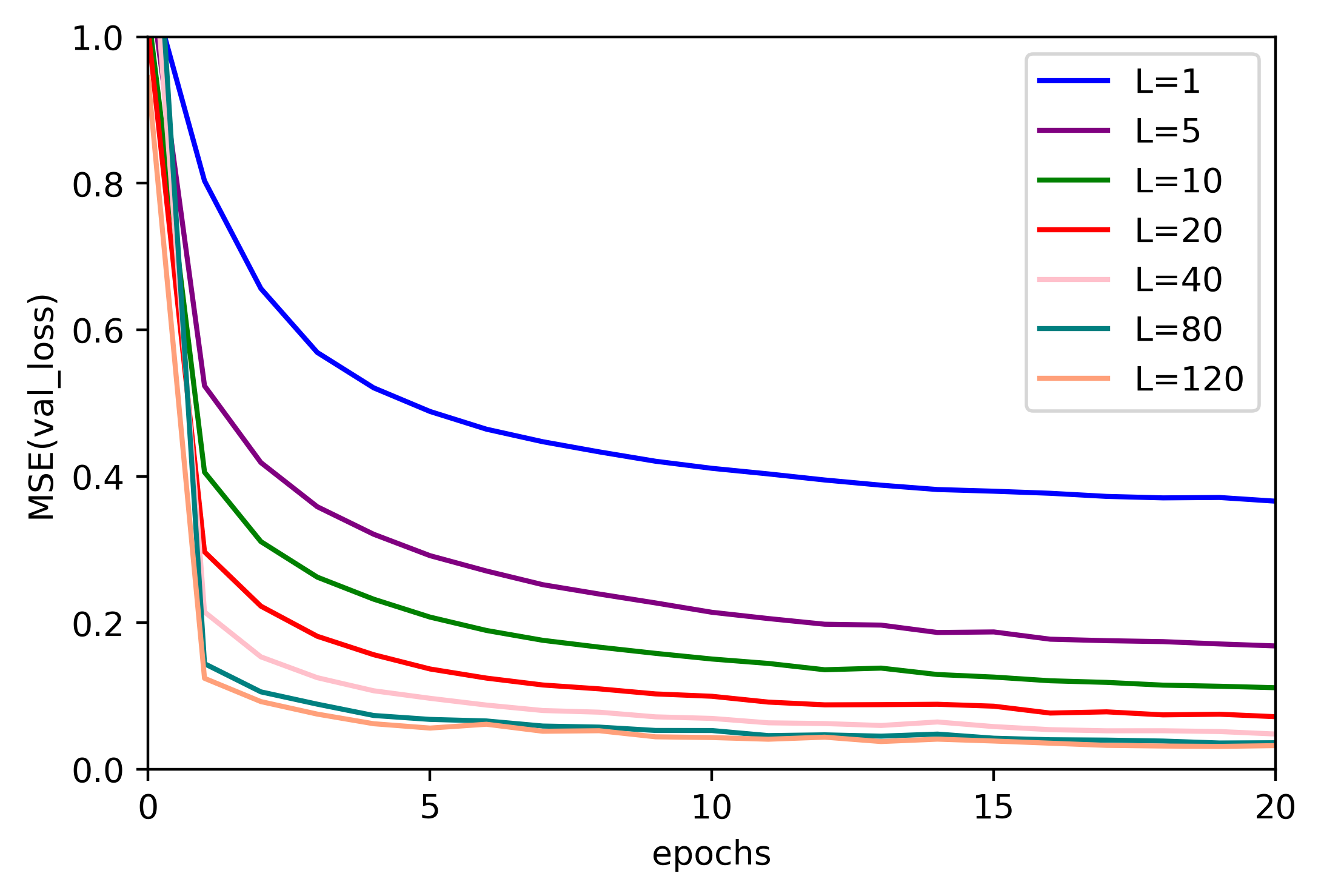}
    \caption{The performance of our model on replicating the output of a un-partitioned VQC when $\theta$ is fixed to that of the un-partitioned model, and only the parameters meant to replicate the partitioning process are trained. The amount of terms included is varied according to the Legend in the top right.}
    \label{fig:synth data theta only}
\end{figure}

The results of our experiment are shown in figure \ref{fig:synth data theta only}. The benefits of increasing $L$ are more apparent than in the digit recognition experiment, we can see better approximations being made at higher $L$. For some applications more accuracy might be required, it seems increasing $L$ further will continue improve this accuracy. Noteably all $L$ considered are orders of magnitude below the amount of terms or circuits needed to apply the existing partitioning schemes. The final mean squared error for unseen data averaged over 5 random data sets is presented in table \ref{tab: no lambda RPM}.
\begin{table}
    \begin{center}
        \begin{tabular}{c|c}
            L & Final validation MSE \\
            \hline
            1  & 0.366     \\
            5  & 0.168     \\
            10 & 0.111     \\
            20 & 0.0717    \\
            40 & 0.0481    \\
            80 & 0.0362    \\
            120& 0.0322    \\
            \hline \\
            Neural Net & 0.424 \\
        \end{tabular}
    \end{center}
    \caption{Results of experiment comparing the validity of our assumptions. We fix the value's of $\theta$ and set the 
    reduced partition model to learn the output of a larger PQC, to simulate the larger PQC exactly using \cite{bravyi2016} would require $L=16,777,216$, the model is only able to select which parameters to put on the gates resulting from the partition. Different values of the hyperparameter $L$ are listed for comparison. A neural network is also trained and performs poorly. Results are averaged over 5 runs.}
    \label{tab: no lambda RPM}
\end{table}
    
Where we have included a neural network with a single dense hidden layer of 256 neurons for comparison purposes, other neural network architectures (1 and 2 hidden layers were tried, with 64 and 256 neurons per layer for each) were tried without meaningful improvement, although it is possible that with thorough tuning these architectures or others could be made to perform strongly.

The previous results are sufficient to show that the training of just the parameters $\bm \zeta$ and $\bm \lambda$ can lead to models with substantially fewer terms, $L$, while still sufficiently approximating the full circuit in this instance. This approximation was achieved with just the training of $\bm \zeta$ and $\bm \lambda$, while fixing the $\bm \theta$ to those that were used to generate the data. However it is not clear, a-priori, that the reduced model should use the same $\bm \theta$ parameters to best mimic the full model. We now allow $\bm \theta$ to deviate from that of the generating PQC, the resulting mean squared error for unseen data after 20 epochs is presented in table \ref{tab:free RPM}.

\begin{table}
    \begin{center}
        \begin{tabular}{c|c}
            L & Final validation MSE \\
            \hline
            1  & 0.176     \\
            5  & 0.112     \\
            10 & 0.0855    \\
            20 & 0.0600    \\
            40 & 0.0434    \\
            80 & 0.0334    \\
            120& 0.0289    \\
            \hline \\
            Neural Net & 0.424 \\
        \end{tabular}
    \end{center}
    \caption{Mean squared error of the reduced partition model on learning the output of a larger PQC, unlike table \ref{tab:handwriting} all parameters are now free and we can directly test the RPM's capabilities on this task. Different values of the hyperparameter $L$ are listed for comparison. A neural network is also trained and performs very poorly. Results averaged over 5 runs.}
    \label{tab:free RPM}
\end{table}
    
This improvement in performance is unsurprising as the unrestricted model includes the hypothesis of the model without training $\bm \theta$, however it was not clear before the experiment that the model would be able to find this higher performance, as the introduction of more parameters may have created too many local optima for efficient optimisation. On the other hand we may have expected a larger increase in performance,  as $\bm \theta$ makes up the majority of parameters, we should expect releasing $\bm \theta$ to correspond to a big increase in performance. The lack of this increase could be taken as weak evidence that our approximation (that a smaller set $L$ can approximate the output of the whole circuit) to be relatively accurate in this case, even without retraining $\bm \theta$.

Finally we use the synthetic data set as a training set for our model, with random initialisation of weights. This third experiment allows us to test our models performance on a task which a classical algorithm (the neural network) performs poorly on, without prior knowledge of good parameters.

\begin{center}
    \begin{tabular}{c|c}
        L & Final validation MSE \\
        \hline
        1  & 0.183     \\
        5  & 0.113     \\
        10 & 0.0944    \\
        20 & 0.0703    \\
        40 & 0.0540    \\
        80 & 0.0357    \\
        120& 0.0362    \\
        \hline \\
        Neural Net & 0.424 \\
    \end{tabular}
\end{center}
    
These performances are strong and comparable to the previous two experiments, where $\bm \theta$ was given, showing that our model performs well on this task, much better than the neural network we compare it to. This final experiment is an excellent demonstration of our model as it would be deployed, and demonstrates that it can learn a non-trivial task where a higher number of qubits would naively be required.

\subsection{Open Source Code}
Open source code is provided at \hyperlink{https://github.com/simon-marshall/High_Dimensional_Quantum_Machine_Learning_Public}{https://github.com/simon-marshall}
in an ipython notebook format with an MIT licence to allow for easy reuse. 2 versions of the code are provided, one to replicate either experiment.

\section{Conclusion and future work}
In this work we applied previously developed circuit cutting techniques to parameterised quantum circuits. While it is obvious that this naive approach used too many circuit evaluations to be computationally practical we noted there may exist a smaller set of circuits which would sufficiently approximate the original circuit, although we speculate that finding it would itself be computationally intractable even if it did exist. Instead we proposed a new model based on the relaxation of fixed gates into parameterised gates, such that all circuits were identical up to the weights of these newly parameterised gates. We showed our models hypothesis class contained the relevant unparameterised hypothesis class, that its generalisation error was well behaved and then went on to test it experimentally. The first experiment showed the model was capable of tackling large problem sizes (handwriting). We also tested the ability of a parameterised subset of circuits of the partition to approximate the full unpartitioned output of a random circuit and found a very satisfying approximation, although a larger amount of terms was needed than with the handwriting task, suggesting a link between the problem and the number of terms needed to achieve a given accuracy.

Further work is needed in establishing how many terms ($L$) might be required for any given task, and what factors influence this requirement. Future work could also focus around the application of this model, testing it out on larger cutting edge problems, or on achieving higher accuracy. Improvements to the model could come from a development of a robust training procedure to avoid barren plateaus (Section \ref{Implementation}) or from integrating our work with some of the excellent work already done on improving divide and conquer schemes (Section \ref{Rel Work}). Our work has opened the door for experimentation with much larger ``partially quantum" models both implicitly as we have done here, but potentially explicitly, integrating more classical resources into a quantum machine learning setting.

\section{Acknowledgements}

The authors would like to thank Matthias Caro for his insights, particularly on the link between different complexity measures and Elies Gil-Fuster. SCM thanks Radoica Draškić and Yash Patel for their useful discussion. The authors thank an anonymous reviewer, Andrea Skolik, Elies Gil-Fuster and Charles Moussa for helpful comments. VD and SCM acknowledge the support by the project NEASQC funded from the European Union’s Horizon 2020 research and innovation programme (grant agreement No 951821). VD and SCM also acknowledge partial funding by an unrestricted gift from Google Quantum AI. VD and CG were supported by the Dutch Research Council (NWO/OCW), as part of the Quantum Software Consortium programme (project number 024.003.037). This work was also supported by the Dutch National Growth Fund (NGF),as part of the Quantum Delta NL programme.

\bibliographystyle{plainnat} 
\bibliography{biblo}

\appendix
\section{Example of increased complexity after the removal of a 2-qubit gate}
\label{appendix}
In this appendix we will see that in some highly manufactured cases removing a two qubit gate can lower expressivity.

Consider the following circuit:

\begin{figure}[H]
    \centering
    \includegraphics[width=6cm]{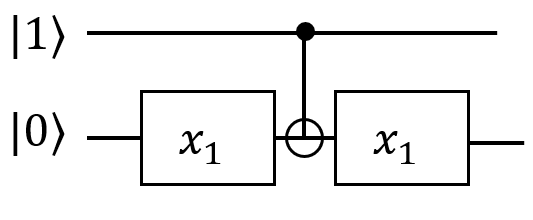}
\end{figure}

Where the encoding gates are Ry. With the CNOT in place the state vector at the end will always be $\ket{11}$, regardless of the input vector. Removing the CNOT entirely produces the following circuit:

\begin{figure}[H]
    \centering
    \includegraphics[width=6cm]{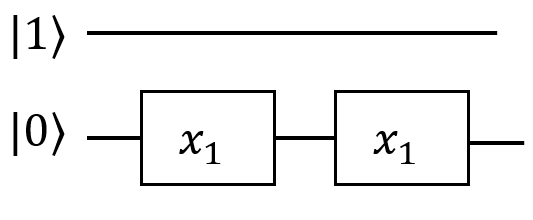}
\end{figure}

With the CNOT removed the final state is $\cos(x_1)\ket{10}+\sin(x_1)\ket{11}$, which depends on the input, $x_1$. In this way it is clear how a circuit making use of this (rather pointless) gate would increase in expressivety when a CNOT is removed. No one would realistically propose this circuit, but its existence prevents blanket statements about the effect of removing two qubit gates on complexity and therefore generalisation.

\section{Proofs of theorems}
Throughout this paper we have stated results without proof. Here we provide proofs of each of our theorems that clear up technical details.

\subsection{Proof of Theorem \ref{partitioned model}}

\begin{theorem*}[Partitioned model]
For every function $f_{\bm \theta} \in \mathcal{F}_{\mathbb{U}, M}$ and qubit partition $\{B_k\}_{k \in [K]}$ with observable $M = \bigotimes_{k \in [K]}M_k$, there exists a set of coefficients $\{c_i\}_{i \in [T]}$ and unitaries $\tilde{\mathbb{U}} = \{U^{i, k}(\bm \theta, \bm  x), U'^{i, k}(\bm \theta, \bm x)\}_{i \in [T], \text{ } k\in [K]}$ (where each $U^{i,k}$ and $U'^{i,k}$ acts on $n_k$ qubits) which can be combined in a function:
\begin{equation}
     \tilde{f}_{\bm \theta}(\bm x) = \sum_{i = 1}^{T} c_i \prod_{k=1}^K \bra{0}U'^{i, k\dagger}(\bm \theta, \bm x)M_{k}U^{i, k}(\bm \theta, \bm x)\ket{0},
\end{equation}
such that $ f_{\bm \theta}(\bm x) = \tilde{f}_{\bm \theta, M}(\bm x)$ for every $\bm \theta, \bm x$. 
For arbitrary gates the number of terms $T$ grows as $16^r$, where r is the number of gates across the partition, but for cut gates with known Schmidt number $T$ is the product of the Schmidt number squared of each cut gate.
\end{theorem*}

\begin{proof}

For any given circuit, $U$, \cite{bravyi2016} provides a decomposition of the form:
$$
U = \sum_{i = 1}^{T} a_i \prod_{k=1}^KU^{i, k}
$$
by writing two qubit gates in the schmidt decomposition. Two qubit gates have schmidt rank of between 2 and 4 \cite{balakrishnan2011operator}, thus any two qubit can be expressed as a sum of 4 tensor product single qubit gates. To express an expectation value in this form requires decomposition on both $U$ and $U^\dagger$, using the linearity of the the expectation value we arrive at $16^r$ inner products for $r$ two qubit gates. Applying this scheme to every element of $\tilde{\mathbb{U}}$ gives us the final statement.
\end{proof}

\subsection{Proof of Lemma \ref{lemma:reduced}}

\begin{lemma*}
For every $\tilde{f}_{\bm \theta} \in \mathcal{F}^L_{\mathbb{U}, M}$, there exists a set of unitaries $\{U^{k}(\bm \theta, \bm x, \bm \zeta)\}_{k\in [K]}$ and parameters $\bm \lambda, \bm \zeta$ defining a function:
\begin{align*}
     &\bar{f}_{\bm \theta, \bm \zeta, \bm \lambda}(\bm x) = \\&\sum_{i \in [L]} {\lambda_i} \prod_{k\in[K]} \bra{0}U^{k \dagger}(\bm \theta,\bm  x, \bm{\zeta}_{i,k})M_{k}U^{k}(\bm \theta,\bm  x, \bm{\zeta}_{i,k+K})\ket{0},
\end{align*}
\noindent
such that $\tilde{f}_{\bm \theta}(\bm x) = \bar{f}_{\bm \theta, \bm \zeta, \bm \lambda}(\bm x)$ for every $\bm \theta,\bm  x$ and for every observable that can be written as tensor product on the elements of the partition $M = \bigotimes_k M_k$. 
\end{lemma*}
\begin{proof}

This is a simple extension of the last theorem. For fixed $i,k$ we have to find $U^{k}(\bm \theta,\bm  x, \bm{\zeta}_{i,k})$ such that
$$
U^{k}(\bm \theta,\bm  x, \bm{\zeta}_{i,k}) = U^{i, k}(\bm \theta,\bm  x) 
$$

for some $\zeta_{i,k}$. We know that the cutting scheme in \cite{bravyi2016} replaces the site of removed two qubit gates with single qubit unitaries, parameterising the difference between the unitaries proves the result.
    
\end{proof}

\subsection{Proof of Theorem \ref{thm:inclusion}}

\begin{theorem*}
For any PQC hypothesis class $\mathcal{F}_{\mathbb{U}, M}$, the $L$-subset partition hypothesis class $\mathcal{F}^L_{\mathbb{U}, M}$ is included in the hypothesis class of the reduced $L$-subset partition model $\overline{\mathcal{F}^L_{\mathbb{U}, M}}$, i.e., 
\begin{equation}
    \mathcal{F}^L_{\mathbb{U}, M} \subset \overline{\mathcal{F}^L_{\mathbb{U}, M}}
\end{equation}
\end{theorem*}
\begin{proof}

By lemma \ref{lemma:reduced} we know that for all $\tilde{f}_{\bm \theta} \in \mathcal{F}^L_{\mathbb{U}, M}$ there exists $\bm \zeta$ selecting a function $\bar{f}_{\bm \theta, \bm \zeta, \bm \lambda} \in \overline{\mathcal{F}^L_{\mathbb{U}, M}}$ such that $\bar{f}_{\bm \theta, \bm \zeta, \bm \lambda}(x)=\tilde{f}_{\bm \theta}(x) \forall x$.
    
\end{proof}

\end{document}